\documentclass[12pt]{article}

\usepackage{hyperref}

\usepackage{amsmath} 
\usepackage{amssymb}
\usepackage{amsfonts}
\usepackage{amsthm}
\usepackage{graphicx}

\newtheorem{theorem}{Theorem}
\newtheorem{lemma}{Lemma}
\newtheorem{corollary}{Corollary}

\theoremstyle{definition}

\newtheorem{definition}{Definition}

\begin{document}
	\begin{center}
		\Large
	\textbf{Time-convolutionless master equations for composite open quantum systems}
	
		\large 
			\textbf{A.Yu. Karasev}\footnote{Faculty of Physics, Lomonosov Moscow State University,
				Leninskie Gory, Moscow 119991, Russia\\ E-mail:\href{mailto:artem.karasev.01@gmail.com}{artem.karasev.01@gmail.com}}
		\textbf{A.E. Teretenkov}\footnote{Department of Mathematical Methods for Quantum Technologies, Steklov Mathematical Institute of Russian Academy of Sciences,
			ul. Gubkina 8, Moscow 119991, Russia\\ E-mail:\href{mailto:taemsu@mail.ru}{taemsu@mail.ru}}
		\end{center}
		
			\footnotesize
			In this work we consider the master equations for composite open quantum systems. We provide purely algebraic formulae for terms of perturbation series defining such equations. We also give conditions under which the Bogolubov-van Hove limit exists and discuss some corrections to this limit. We present an example to illustrate our results. In particular, this example shows, that inhomogeneous terms in time-convolutionless master equations can vanish after reservoir correlation time, but lead to renormalization of initial conditions at such a timescale.
			\normalsize

	\section{Introduction}

The usual setup of open quantum systems theory assumes unitary dynamics of a system and a heat bath \cite{Breuer2002}. In this work we consider dynamics, where the system of interest and the reservoir initially constitute a composite open quantum system instead. Such an approach was not so widely discussed in literature, but has attracted more and more attention in recent works  \cite{Saideh2020, Finkelstein-Shapiro2020, Regent2023, Regent2023a}. It was also discussed in a setup very different from the usual open quantum systems one  \cite{Arefeva2017}. This approach seems to be natural when a system in a reservoir is composite, i.e. consists of subsystems, and we are only interested in one of the subsystems. Moreover, some authors hypothesize that even fundamental dynamics of our world could be described like open quantum system one rather than just unitary one \cite{Weinberg2016}.

We use the Nakajima-Zwanzig projection approach \cite{Nakajima1958, Zwanzig1960} in our work, both because it is widely used in open quantum systems theory, and because it allows using different projection, i.e. different approaches to identification of the system of interest inside the composite system, in a uniform manner. We assume a small coupling between the system of interest and the reservoir. We are looking for the approach which in principle can work in all the orders of perturbation theory, so in Section \ref{se:TCLequation} we derive time-convolutionless master equations. Such equations were introduced in \cite{Fulinski1967} and are widely used in open quantum systems theory \cite{Shibata1977, Breuer1999, Breuer2001, Semin2020}. We also use an inhomogenious term to take into account the non-factorizible initial conditions \cite{Chang1993}. Moreover, time-nonlocal equations are also useful in certain cases \cite{Filippov2022}.

Our work is organized from the general to the particular. We are concerned with the reduced dynamics of  the system of interest inside the composite system in the interaction picture. But the perturbative formulae from Section \ref{se:TCLequation} are general and do not take into account the origin of initial time-dependent equations which we project. So in Section \ref{sec:Explicit} we take into account that such initial equations occur in the interaction picture. This allows us to write more explicit formulae. Namely, the expression defining the time-convolutionless master equation becomes purely algebraic.

In our opinion the identification of a system inside a composed system is not physical by itself without some dynamical conditions. Typically such a condition is Markovianity of system dynamics. From the mathematical point of view Markovian dynamics occurs \cite{Davies1974, Accardi1990, Accardi2002} in the Bogolubov-van Hove limit \cite{Bogoliubov1946, VanHove1954}. So it is natural to ask when the Bogolubov-van Hove limit exists in our case. In Section \ref{sec:relaxCond} we introduce some sufficient conditions under which such a limit exists and discuss some  corrections to this limit.

To illustrate our results we consider an example in Section \ref{sec:example}. In our opinion, the most interesting feature of this example is that inhomogeneous terms in time-convolutionless master equations vanish in the Bogolubov-van Hove perturbation theory, but lead to renormalization of initial conditions.

In Conclusions we summarize our work and suggest some possible directions for further study.  

\section{Time-convolutionless equations}\label{se:TCLequation}

In this section we derive the time-convolutionless master equations. We consider only finite-dimensional matrices and maps throughout the work. We refer to linear maps between finite-dimensional matrices as superoperators.  Usually, the time-convolutionless master equations  are defined perturbatively in terms of Kubo - van Kampen cumulants \cite{Kubo1963, VanKampen1974, VanKampen1974a, Chaturvedi1979, Shibata1980}. But here we represent them in terms of non-commutative analog of moments similar to \cite{Nestmann2019, Teretenkov2022} and then obtain explicit formulae for them in the next section. Namely, we obtain the following theorem describing perturbation series for the time-convolutionless master equation.

\begin{theorem}\label{th:masterEquation}
	Let $ \rho(t)$ be a matrix-valued function of $ t\in  [t_0, +\infty)$ which satisfies an ordinary differential equation
	\begin{equation}\label{eq:basicDiffEq}
		\frac{d}{dt} \rho(t) = \lambda \mathcal{L}(t)\rho(t),
	\end{equation}
	where $\mathcal{L}(t)$ is a superoperator-valued function, which is continuous for  $ t\in  [t_0, +\infty)$,  and $\lambda$ is a parameter,  $\mathcal{P}$ is an idempotent superoperator $ \mathcal{P}= \mathcal{P}^2 $ and $ \mathcal{Q} \equiv I - \mathcal{P}$, where $I$ is an identity superoperator. Then $\mathcal{P} \rho(t)$ satisfies the time-convolutionless master equation
	\begin{equation}\label{eq:masterEquation}
		\frac{d}{dt} \mathcal{P} \rho(t) = \mathcal{K}(t) \mathcal{P} \rho(t) + \mathcal{I}(t) \mathcal{Q} \rho(t_0) , 
	\end{equation}
	where $\mathcal{K}(t)$ has the following asymptotic expansion at $\lambda \rightarrow 0 $
	\begin{equation}\label{eq:expansionOfK}
		\mathcal{K}(t) = \sum_{n=1}^{\infty} \lambda^n \mathcal{K}_n(t)
	\end{equation}
	with
	\begin{equation}\label{eq:coeffKn}
		\mathcal{K}_n(t) =  \sum_{q=0}^{n-1} (-1)^q \sum_{\sum_{j=0}^q k_j =n, k_j \geqslant 1}  \dot{\mathcal{M}}_{k_0}(t) \mathcal{M}_{k_1}(t)  \ldots \mathcal{M}_{k_{q}}(t),
	\end{equation}
	i.e. the inner sum runs over all possible compositions, 
	\begin{equation}\label{eq:momentsDef}
		\mathcal{M}_{k}(t) = \int_{t_0}^t d t_1  \ldots  \int_{t_0}^{t_{k-1}} d t_k  \mathcal{P}  \mathcal{L}(t_1) \ldots \mathcal{L}(t_k)\mathcal{P} , \qquad \dot{\mathcal{M}}_{k}(t) \equiv \frac{d}{dt} \mathcal{M}_{k}(t),
	\end{equation}
	and  $\mathcal{I}(t)$ has the following asymptotic expansion at $\lambda \rightarrow 0 $
	\begin{equation}\label{eq:expansionOfI}
		\mathcal{I}(t)  = \sum_{n=1}^{\infty} \lambda^n \mathcal{I}_n(t)
	\end{equation}
	with
	\begin{equation*}
		\mathcal{I}_n(t) = \dot{\tilde{\mathcal{M}}}_{n}(t) +  \sum_{q=1}^{n-1} (-1)^q \sum_{\sum_{j=0}^q k_j =n, k_j \geqslant 1}  \dot{\mathcal{M}}_{k_0}(t) \mathcal{M}_{k_1}(t)  \ldots \mathcal{M}_{k_{q}-1}(t)\tilde{ \mathcal{M}}_{k_{q}}(t),
	\end{equation*}
	\begin{equation}\label{eq:comomentsDef}
		\tilde{\mathcal{M}}_k(t) = \int_{t_0}^t d t_1  \ldots  \int_{t_0}^{t_{k-1}} d t_k  \mathcal{P}  \mathcal{L}(t_1) \ldots \mathcal{L}(t_k)\mathcal{Q}, \qquad \dot{\tilde{\mathcal{M}}}_{k}(t) \equiv  \frac{d}{dt} \tilde{\mathcal{M}}_{k}(t).
	\end{equation}
\end{theorem}

Remark that $  \mathcal{I}_n(t) $ has the same form as $ \mathcal{K}_n(t)$, but with  $ \mathcal{Q} $ instead of last  $ \mathcal{P} $ in each term, which coincides with \cite[Section 9.2.4]{Breuer2002}.

First of all let us prove the following lemma.	 
\begin{lemma}
	\label{lem:inverseMatrixSeries}
	Let $\mathcal{A}_k$ be superoperators, then the following asymptotic expansion holds
	\begin{equation}\label{eq:inverseMatrixSeries}
		\left(I + \sum_{k=1}^{\infty} \lambda^k \mathcal{A}_k\right)^{-1} 
		=  \sum_{n=0}^{\infty} \lambda^n \sum_{q=0}^n (-1)^q \sum_{\sum_{j=1}^q k_j =n, k_j \geqslant 1}  \mathcal{A}_{k_1} \ldots  \mathcal{A}_{k_q}, \qquad \lambda \rightarrow 0.
	\end{equation}
\end{lemma}

\begin{proof}
	By direct calculation we have
	\begin{align*}
		\left(I + \sum_{k=1}^{\infty} \lambda^k \mathcal{A}_k\right)^{-1}   &= \sum_{q=0}^{\infty}(-1)^q \left(\sum_{k=1}^{\infty} \lambda^k \mathcal{A}_k\right)^{q}
		\\
		&= \sum_{q=0}^{\infty}(-1)^q \sum_{n=1}^{\infty} \lambda^n  \sum_{\sum_{j=1}^q k_j =n, k_j \geqslant 1} \mathcal{A}_{k_1}  \ldots  \mathcal{A}_{k_q} \\ 
		&=  \sum_{n=0}^{\infty} \lambda^n  \sum_{\sum_{j=1}^q k_j =n, k_j \geqslant 1} (-1)^q \mathcal{A}_{k_1}  \ldots  \mathcal{A}_{k_q} 
	\end{align*}
	Thus, we obtain \eqref{eq:inverseMatrixSeries}.
\end{proof}

\begin{proof}[Proof of theorem \ref{th:masterEquation}]
	The solution of Equation~\eqref{eq:basicDiffEq} with given $\rho(t_0)$ is unique \cite[Theorem 5.2]{Coddington1955} and can be represented  in the form 
	\begin{equation*}
		\rho(t) = \mathcal{U}_{t_0}^t\rho(t_0),
	\end{equation*}
	where $\mathcal{U}_{t_0}^t$ is a superoperator-valued function of $t$ which is a unique solution of the Cauchy problem
	\begin{equation}\label{eq:propEq}
		\frac{d}{dt}\mathcal{U}_{t_0}^t = \lambda \mathcal{L}(t) \mathcal{U}_{t_0}^t, \qquad \mathcal{U}_{t_0}^{t_0} = I.
	\end{equation}
	Applying the projector $\mathcal{P}$ on both sides of this equation we obtain
	\begin{align*}
		\frac{d}{dt} \mathcal{P}\mathcal{U}_{t_0}^t  &= \frac{d}{dt}  \mathcal{P}\mathcal{U}_{t_0}^t \mathcal{P} +  \frac{d}{dt} \mathcal{P}\mathcal{U}_{t_0}^t \mathcal{Q} = \left(\frac{d}{dt}  \mathcal{P}\mathcal{U}_{t_0}^t \mathcal{P}\right) (\mathcal{P}\mathcal{U}_{t_0}^t \mathcal{P})^{(-1)} \mathcal{P}\mathcal{U}_{t_0}^t \mathcal{P} +  \frac{d}{dt} \mathcal{P}\mathcal{U}_{t_0}^t \mathcal{Q} \\
		&= \left(\frac{d}{dt}  \mathcal{P}\mathcal{U}_{t_0}^t \mathcal{P}\right) (\mathcal{P}\mathcal{U}_{t_0}^t \mathcal{P})^{(-1)} \mathcal{P}\mathcal{U}_{t_0}^t  -  \left(\frac{d}{dt}  \mathcal{P}\mathcal{U}_{t_0}^t \mathcal{P}\right) (\mathcal{P}\mathcal{U}_{t_0}^t \mathcal{P})^{(-1)} \mathcal{P}\mathcal{U}_{t_0}^t \mathcal{Q} +  \frac{d}{dt} \mathcal{P}\mathcal{U}_{t_0}^t \mathcal{Q},
	\end{align*}
	where for some supeoperator $\mathcal{A}$
	\begin{equation}\label{eq:pseudoInverse}
		\mathcal{A}^{(-1)}
	\end{equation}
	denotes the pseudoinverse such that $\mathcal{A}^{(-1)} \mathcal{A} = \mathcal{P}$, $\mathcal{A}^{(-1)} \mathcal{Q} = \mathcal{Q} \mathcal{A}^{(-1)}  = 0$, where in general $ \mathcal{P} $ is a projector to the image of $ \mathcal{A} $ and $  \mathcal{Q}  $ to it kernel.
	
	Taking into account the idempotent property $ \mathcal{P}^2= \mathcal{P}  $ and $ \mathcal{Q}^2 = (I- \mathcal{P})^2 = I - 2 \mathcal{P} + \mathcal{P}^2= I - \mathcal{P} = \mathcal{Q} $ we obtain
	\begin{equation*}
		\frac{d}{dt} \mathcal{P}\mathcal{U}_{t_0}^t  = \mathcal{K}(t) \mathcal{P}\mathcal{U}_{t_0}^t  + \mathcal{I}(t) \mathcal{Q},
	\end{equation*}
	with
	\begin{equation}\label{eq:Kdef}
		\mathcal{K}(t) \equiv  \left(\frac{d}{dt}  \mathcal{P}\mathcal{U}_{t_0}^t \mathcal{P}\right) (\mathcal{P}\mathcal{U}_{t_0}^t \mathcal{P})^{(-1)}
	\end{equation}
	and
	\begin{equation}\label{eq:Idef}
		\mathcal{I}(t) \equiv -  \left(\frac{d}{dt}  \mathcal{P}\mathcal{U}_{t_0}^t \mathcal{P}\right) (\mathcal{P}\mathcal{U}_{t_0}^t \mathcal{P})^{(-1)} \mathcal{P}\mathcal{U}_{t_0}^t \mathcal{Q} +  \frac{d}{dt} \mathcal{P}\mathcal{U}_{t_0}^t \mathcal{Q}  =  \frac{d}{dt} \mathcal{P}\mathcal{U}_{t_0}^t \mathcal{Q} -  \mathcal{K}(t)  \mathcal{P}\mathcal{U}_{t_0}^t \mathcal{Q} .
	\end{equation}
	
	Expanding  $\mathcal{U}_{t_0}^t $ defined by Cauchy problem \eqref{eq:propEq} asymptotically in the Dyson series we obtain for $\lambda \rightarrow 0$
	\begin{equation*}
		\mathcal{U}_{t_0}^t = \sum_{k=0}^{\infty} \lambda^k \int_{t_0}^t d t_1  \ldots  \int_{t_0}^{t_{k-1}} d t_k \mathcal{L}(t_1) \ldots \mathcal{L}(t_k).
	\end{equation*}
	
	Hence, we have
	\begin{equation*}
		\mathcal{P}\mathcal{U}_{t_0}^t\mathcal{P} = \mathcal{P} +\sum_{k=1}^{\infty} \lambda^k \mathcal{M}_{k}(t), \qquad 	\frac{d}{dt}\mathcal{P}\mathcal{U}_{t_0}^t\mathcal{P} = \sum_{k=1}^{\infty} \lambda^k \dot{\mathcal{M}}_{k}(t).
	\end{equation*}
	Then using Lemma \ref{lem:inverseMatrixSeries} from \eqref{eq:Kdef} we obtain
	\begin{align*}
		\mathcal{K}(t) =\frac{d}{dt}\mathcal{P}\mathcal{U}_{t_0}^t \mathcal{P} \left(\mathcal{P}\mathcal{U}_{t_0}^t \mathcal{P}\right)^{(-1)} &= \sum_{k=0}^{\infty} \lambda^k \dot{\mathcal{M}}_{k}(t)  \left(\sum_{m=0}^{\infty} \lambda^m \mathcal{M}_{m}(t) \right)^{(-1)} \\
		&=  \sum_{k=0}^{\infty} \lambda^k \dot{\mathcal{M}}_{k}(t) \left(  \sum_{n=0}^{\infty} \lambda^n \sum_{q=0}^n (-1)^q \sum_{\sum_{j=1}^q k_j =n, k_j \geqslant 1}  \mathcal{M}_{k_1}(t) \ldots  \mathcal{M}_{k_q}(t)\right)\\
		&= \sum_{n=0}^{\infty} \lambda^n \sum_{q=0}^{n-1} (-1)^q \sum_{\sum_{j=0}^q k_j =n, k_j \geqslant 1}  \dot{\mathcal{M}}_{k_0}(t) \mathcal{M}_{k_1}(t)  \ldots \mathcal{M}_{k_{q}}(t)
	\end{align*}
	Thus, we obtain \eqref{eq:expansionOfK}.
	
	Similarly, we have
	\begin{equation*}
		\mathcal{P}\mathcal{U}_{t_0}^t\mathcal{Q} = \sum_{k=1}^{\infty} \lambda^k \tilde{\mathcal{M}}_{k}(t), \qquad 	\frac{d}{dt}\mathcal{P}\mathcal{U}_{t_0}^t\mathcal{Q} = \sum_{k=1}^{\infty} \lambda^k \dot{\tilde{\mathcal{M}}}_{k}(t),
	\end{equation*}
	then \eqref{eq:Idef} leads to
	\begin{align*}
		\mathcal{I}(t) 
		&= \sum_{k=1}^{\infty} \lambda^n \dot{\tilde{\mathcal{M}}}_{n}(t) - \left(\sum_{n=0}^{\infty} \lambda^n \sum_{q=0}^{n-1} (-1)^q \sum_{\sum_{j=0}^q k_j =n, k_j \geqslant 1}  \dot{\mathcal{M}}_{k_0}(t) \mathcal{M}_{k_1}(t)  \ldots \mathcal{M}_{k_{q}}(t)\right) \left(\sum_{k=1}^{\infty} \lambda^k \dot{\tilde{\mathcal{M}}}_{k}(t)\right)\\
		&= \sum_{k=1}^{\infty} \lambda^n \left(\dot{\tilde{\mathcal{M}}}_{n}(t) +  \sum_{q=1}^{n-1} (-1)^q \sum_{\sum_{j=0}^q k_j =n, k_j \geqslant 1}  \dot{\mathcal{M}}_{k_0}(t) \mathcal{M}_{k_1}(t)  \ldots \mathcal{M}_{k_{q}-1}(t)\tilde{ \mathcal{M}}_{k_{q}}(t)\right)
	\end{align*}
	Thus, we obtain \eqref{eq:expansionOfI}.
\end{proof}

Let us provide explicitly  several first terms $ \mathcal{K}_n(t) $  in terms of $ \mathcal{M}_n(t) $ to illustrate formula \eqref{eq:coeffKn}.
\begin{equation}\label{eq:K1M1K2M2}
	\mathcal{K}_1(t) = \dot{\mathcal{M}}_{1}(t), \qquad \mathcal{K}_2(t) = \dot{\mathcal{M}}_{2}(t) -  \dot{\mathcal{M}}_{1}(t) \mathcal{M}_{1}(t),
\end{equation}
\begin{equation*}
	\mathcal{K}_3(t) = \dot{\mathcal{M}}_{3}(t) -  \dot{\mathcal{M}}_{2}(t) \mathcal{M}_{1}(t) -  \dot{\mathcal{M}}_{1}(t) \mathcal{M}_{2}(t) + \dot{\mathcal{M}}_{1}(t) \mathcal{M}_{1}(t) \mathcal{M}_{1}(t).
\end{equation*}

Remark that we are interested in the dynamics $\mathcal{P} \rho(t)$, but there are known methods \cite{Trushechkin2019} which allow one to calculate $\mathcal{Q} \rho(t)$ perturbatively using the perturbative result for  $\mathcal{P} \rho(t)$.

\section{Explicit formulae in the case of interaction picture}
\label{sec:Explicit}
Similar to the usual open quantum systems setup \cite[Section 9.1]{Breuer2002} it is natural to start with the equation with the time-independent generator of the form
\begin{equation*}
	\frac{d}{dt}\tilde{\rho}(t) = (\mathcal{L}_0 + \lambda \mathcal{L})\tilde{\rho}(t),
\end{equation*}
where $ \mathcal{L}_0  $ and $\mathcal{L}$ are Gorini--Kossakowski--Sudarshan--Lindblad (GKSL) generators \cite{Gorini1976, Lindblad1976} in our case.
And then one can  move to the interaction picture
\begin{equation*}
	\rho(t) \equiv e^{ -\mathcal{L}_0 t }\tilde{\rho}(t),
\end{equation*}
which leads to the equation of the form \eqref{eq:basicDiffEq} with $ \mathcal{L}(t)  = e^{- \mathcal{L}_0 t}\mathcal{L} e^{\mathcal{L}_0 t} = e^{- \mathfrak{L} t} \mathcal{L} $, where $  \mathfrak{L} = [\mathcal{L}_0, \; \cdot \;] $. Here $  [\mathcal{L}_0, \; \cdot \;]  $ denotes the map from superoperators to superoperators which is defined on superoperator $\mathcal{A}$ by the formula  $  [\mathcal{L}_0, \; \cdot \;] \mathcal{A} =  [\mathcal{L}_0, \mathcal{A}]  $.

To write some explicit algebraic formulae below, we will need the following definition.

\begin{definition}
	Let us define the contraction map as the linear map such that for any superoperators  $\mathcal{A}_1$,  \ldots , $\mathcal{A}_k$ one has
	\begin{equation*}
		\mathfrak{C}(\mathcal{A}_1 \otimes \ldots \otimes \mathcal{A}_k) =  \mathcal{A}_1  \ldots \mathcal{A}_k.
	\end{equation*}
\end{definition}

\begin{lemma}\label{lem:compressionLemma}
	For $  \mathcal{L}(t)  = e^{- \mathfrak{L} t} \mathcal{L}$ for some linear map $ \mathfrak{L}  $ and let us denote by subscript $j$ in $\mathfrak{L}_j$ the fact that this map acts as  $\mathfrak{L}$  in the $j$-th tensor multiplicand and as identity in other ones. Then
	\begin{equation}\label{eq:compressionLemma}
		\mathcal{L}(t_1 ) \ldots \mathcal{L}(t_k) = \mathfrak{C}(e^{- \sum_{j=1}^kt_j \mathfrak{L}_j}  \mathcal{L}^{\otimes k}).
	\end{equation}
\end{lemma}

\begin{proof}
	By direct calculation we have
	\begin{align*}
		\mathcal{L}(t_1 ) \ldots \mathcal{L}(t_k) &=	\mathfrak{C}(\mathcal{L}(t_1 ) \otimes \ldots \otimes  \mathcal{L}(t_k) ) 
		= \mathfrak{C}(e^{- \mathfrak{L} t_1}  \mathcal{L}  \otimes \ldots \otimes  e^{- \mathfrak{L} t_k}  \mathcal{L} )\\
		&= \mathfrak{C}((e^{- \mathfrak{L} t_1}    \otimes \ldots \otimes  e^{- \mathfrak{L} t_k} ) (\mathcal{L} \otimes \ldots \otimes\mathcal{L}) )  = \mathfrak{C}(e^{- \sum_{j=1}^kt_j \mathfrak{L}_j}  \mathcal{L}^{\otimes k}).
	\end{align*}
	Thus, we obtain \eqref{eq:compressionLemma}.
\end{proof}

We also need Lemma A1 from \cite{Teretenkov2022} which takes the following form.

\begin{lemma}\label{lemma:auxiliary}
	The following formula holds 
	\begin{align}
		h_{k}(t; \gamma_1, \ldots, \gamma_k)  &\equiv \int_{0}^{t} d t_1  \ldots  \int_{0}^{t_{k-1}} d t_k e^{ - \sum_{j=1}^k t_j \gamma_j} \nonumber \\
		&= \frac{1}{\prod_{n=1}^k \sum_{j=1}^n\gamma_j}
		+ \sum _{p=1}^k (-1)^p \frac{ e^{-t  \sum_{i=1}^p \gamma_i}}{\left(\prod_{m=1}^p	\sum_{i=m}^p \gamma_i\right) \left(\prod_{n=p+1}^k \sum_{j=p+1}^n \gamma_j\right)}.
		\label{eq:auxiliarylemma}
	\end{align}
	For zero denominators, the right-hand side should be understood as a limit.
\end{lemma}

Let us illustrate this Lemma by examples to clarify its meaning. In particular,
\begin{equation*}
	h_1(t; \gamma_1) = \frac{1 - e^{-\gamma_1 t}}{\gamma_1},
\end{equation*}
where for $ \gamma_1 = 0$ it should be understood as
\begin{equation*}
	h_1(t; 0) \equiv h_1(t; \gamma_1) |_{\gamma_1 \rightarrow 0} = t.
\end{equation*}
Similarly, we have
\begin{equation*}
	h_2(t; \gamma_1, \gamma_2) =
	\begin{cases}
		\frac{1}{\gamma_2}\left(\frac{1 - e^{-\gamma_1 t}}{\gamma_1} - \frac{1 - e^{-(\gamma_1 +\gamma_2)t}}{\gamma_1 +\gamma_2}\right), & \gamma_1 \neq 0, \gamma_2 \neq 0, \gamma_1 \neq - \gamma_2,\\
		- \frac{1- \gamma_2 t - e^{-\gamma_2 t}}{\gamma_2^2}, & \gamma_1 = 0, \gamma_2 \neq 0,\\
		- \frac{1- \gamma_2 t - e^{-\gamma_2 t}}{\gamma_2^2}, & \gamma_2 = - \gamma_1 \neq 0,\\
		\frac{1- (1 + \gamma_1 t) e^{-\gamma_1 t} }{\gamma_1^2}, & \gamma_2 = 0, \gamma_1 \neq 0,\\
		\frac{t^2}{2}, & \gamma_1 = \gamma_2 = 0.
	\end{cases}
\end{equation*}

The explicit purely algebraic expressions for perturbative expansions can be obtained by Theorem~\ref{th:masterEquation} along with the following theorem.
\begin{theorem}\label{th:explicitFormulae}
	For $\mathcal{L}(t)  = e^{- \mathfrak{L} t} \mathcal{L}$  formulae \eqref{eq:momentsDef} and \eqref{eq:comomentsDef} take the form
	\begin{align}
		\mathcal{M}_{k}(t) &=  \mathcal{P} \mathfrak{C} (e^{- t_0 \sum_{j=1}^k\mathfrak{L}_j} h_{k}(t-t_0; \mathfrak{L}_1, \ldots, \mathfrak{L}_k)   \mathcal{L}^{\otimes k})\mathcal{P}, \label{eq:explicitMk}\\
		\tilde{\mathcal{M}}_{k}(t) &=  \mathcal{P} \mathfrak{C} (e^{- t_0 \sum_{j=1}^k\mathfrak{L}_j} h_{k}(t-t_0; \mathfrak{L}_1, \ldots, \mathfrak{L}_k)   \mathcal{L}^{\otimes k})\mathcal{Q}, \label{eq:explicitMkt}\\
		\dot{\mathcal{M}}_{k}(t) &=  \mathcal{P} \mathfrak{C} (e^{ -t \mathfrak{L}_1} e^{- t_0 \sum_{j=2}^k\mathfrak{L}_j} h_{k}(t-t_0; \mathfrak{L}_2, \ldots, \mathfrak{L}_k)   \mathcal{L}^{\otimes k})\mathcal{P}, \label{eq:explicitMkd}\\
		\dot{\tilde{\mathcal{M}}}_{k}(t) &=  \mathcal{P} \mathfrak{C} (e^{ -t \mathfrak{L}_1} e^{- t_0 \sum_{j=2}^k\mathfrak{L}_j} h_{k}(t-t_0; \mathfrak{L}_2, \ldots, \mathfrak{L}_k)   \mathcal{L}^{\otimes k})\mathcal{Q}. \label{eq:explicitMkdt}
	\end{align}
\end{theorem}

\begin{proof}
	By definition \eqref{eq:momentsDef} we have
	\begin{align*}
		\mathcal{M}_{k}(t) &=  \int_{t_0}^t d t_1  \ldots  \int_{t_0}^{t_{k-1}} d t_k  \mathcal{P}  \mathcal{L}(t_1) \ldots \mathcal{L}(t_k)\mathcal{P} \\
		&=  \int_{0}^{t- t_0} d t_1  \ldots  \int_{0}^{t_{k-1}} d t_k  \mathcal{P}  \mathcal{L}(t_1 + t_0) \ldots \mathcal{L}(t_k +t_0)\mathcal{P} .
	\end{align*}
	Due to Lemma \ref{lem:compressionLemma} we obtain
	\begin{align*}
		\mathcal{M}_{k}(t) &=    \int_{0}^{t- t_0} d t_1  \ldots  \int_{0}^{t_{k-1}} d t_k   \mathcal{P} \mathfrak{C}(e^{- \sum_{j=1}^k(t_j +t_0)\mathfrak{L}_j}  \mathcal{L}^{\otimes k})\mathcal{P} \\
		&=    \mathcal{P} \mathfrak{C}( e^{- t_0 \sum_{j=1}^k\mathfrak{L}_j} \int_{0}^{t- t_0} d t_1  \ldots  \int_{0}^{t_{k-1}} d t_k   e^{- \sum_{j=1}^kt_j\mathfrak{L}_j}  \mathcal{L}^{\otimes k})\mathcal{P} .
	\end{align*}
	Using Lemma \ref{lemma:auxiliary} we obtain \eqref{eq:explicitMk}. Similarly, one could obtian \eqref{eq:explicitMkt}--\eqref{eq:explicitMkdt}.
\end{proof}

Let us explicitly write down the terms for $t_0 = 0$ which contribute to the second order expansion of Equation \eqref{eq:masterEquation}
\begin{align}
	&\mathcal{M}_{1}(t) =   \mathcal{P}\left( \frac{1 - e^{-\mathfrak{L} t}}{\mathfrak{L}}  \mathcal{L} \right)\mathcal{P},  \qquad
	\dot{\mathcal{M}}_{1}(t) =   \mathcal{P}\left(  e^{-\mathfrak{L} t} \mathcal{L} \right)\mathcal{P}, \label{eq:M1}\\
	&\tilde{\mathcal{M}}_{1}(t) =   \mathcal{P}\left(  \frac{1 - e^{-\mathfrak{L} t}}{\mathfrak{L}}  \mathcal{L} \right)\mathcal{Q},  \qquad
	\dot{\tilde{\mathcal{M}}}_{1}(t) =   \mathcal{P}\left(  e^{-\mathfrak{L} t} \mathcal{L} \right)\mathcal{Q}, \label{eq:Mt1}
	\\
	&\dot{\mathcal{M}}_{2}(t) =   \mathcal{P}( e^{ -t \mathfrak{L}} \mathcal{L}  ) \left(\frac{1 - e^{-\mathfrak{L} t}}{\mathfrak{L}}  \mathcal{L} \right)\mathcal{P}, \quad \dot{\tilde{\mathcal{M}}}_{2}(t) =   \mathcal{P}( e^{ -t \mathfrak{L}} \mathcal{L}  ) \left(\frac{1 - e^{-\mathfrak{L} t}}{\mathfrak{L}}  \mathcal{L} \right)\mathcal{Q}. \label{eq:M2}
\end{align}

Let us remark that 
\begin{equation*}
	\frac{1 - e^{-\mathfrak{L}_1 t}}{\mathfrak{L}_1} 
\end{equation*}
can be understood as a series
\begin{equation*}
	\frac{1 - e^{-\mathfrak{L}_1 t}}{\mathfrak{L}_1} = \sum_{j=0}^{\infty} (-1)^{j+1} \frac{t^j}{j!} \mathfrak{L}_1^j,
\end{equation*}
so it is defined for degenerate map $\mathfrak{L}_1$. Moreover, one could write
\begin{equation}\label{eq:inverseIndep}
	\frac{1 - e^{-\mathfrak{L}_1 t}}{\mathfrak{L}_1}  = (1 - e^{-\mathfrak{L}_1 t})(\mathfrak{L}_1)^{(-1)},
\end{equation}
where $(\mathfrak{L}_1)^{(-1)}$ can be understood similarly to \eqref{eq:pseudoInverse} as inverse on the image of $\mathfrak{L}_1$ and zero otherwise. But due to the fact that the kernel $ 1 - e^{-\mathfrak{L}_1 t} $ includes the kernel of $\mathfrak{L}_1$ one could interpret $(\mathfrak{L}_1)^{(-1)} \mathcal{A}$ for a superoperator $ \mathcal{A} $  as any solution $\mathcal{X}$ of the equation $\mathfrak{L}_1 \mathcal{X} = \mathcal{A} $. Despite the ambiguity of such a solution for the case $  \mathfrak{L} = [\mathcal{L}_0, \; \cdot \;] $ the whole result $(1 - e^{-\mathfrak{L}_1 t})(\mathfrak{L}_1)^{(-1)} \mathcal{A}$ will not depend on it.

Similarly, we have for example
\begin{equation*}
	\mathcal{M}_{2}(t) =   \mathcal{P}\left( e^{- t_0 (\mathfrak{L}_1 + \mathfrak{L}_2)}\left(\frac{1 - e^{-\mathfrak{L}_1 t}}{\mathfrak{L}_1} - \frac{1 - e^{-(\mathfrak{L}_1 +\mathfrak{L}_2)t}}{\mathfrak{L}_1 +\mathfrak{L}_2}\right) \frac{1}{\mathfrak{L}_2} \mathcal{L} \otimes \mathcal{L} \right) \mathcal{P}.
\end{equation*}

Let us write down the coefficients of the second order asymptotic expansion of the time-convolutionless master equaiton in an even more explicit form.
\begin{corollary}
	For $t_0=0$ and  $\mathcal{L}(t) = e^{- \mathcal{L}_0 t}\mathcal{L} e^{\mathcal{L}_0 t}$ one has
	\begin{equation}\label{eq:K12}
		\mathcal{K}_{1}(t) =\mathcal{P}  \mathcal{L}(t) \mathcal{P}, \qquad \mathcal{K}_{2}(t) =\mathcal{P} \mathcal{L}(t) \mathcal{Q} [\mathcal{L}_0, \; \cdot \;]^{(-1)} (\mathcal{L} - \mathcal{L}(t))\mathcal{P},
	\end{equation}
	\begin{equation}\label{eq:I12}
		\mathcal{I}_{1}(t) =\mathcal{P}  \mathcal{L}(t) \mathcal{Q}, \qquad \mathcal{I}_{2}(t) =  \mathcal{P} \mathcal{L}(t) \mathcal{Q} [\mathcal{L}_0, \; \cdot \;]^{(-1)} (\mathcal{L} - \mathcal{L}(t))\mathcal{Q}.
	\end{equation}
\end{corollary}

\begin{proof}
	Due to formulae \eqref{eq:K1M1K2M2} and \eqref{eq:M1}-\eqref{eq:M2} for $  \mathfrak{L} = [\mathcal{L}_0, \; \cdot \;] $ we have
	\begin{align*}
		\mathcal{K}_{1}(t) =&  \mathcal{P}\left(  e^{-\mathfrak{L} t} \mathcal{L} \right)\mathcal{P}=  \mathcal{P}  \mathcal{L}(t) \mathcal{P} , \\
		\mathcal{K}_{2}(t) = & \mathcal{P}( e^{ -t \mathfrak{L}} \mathcal{L}  ) \left(\frac{1 - e^{-\mathfrak{L} t}}{\mathfrak{L}}  \mathcal{L} \right)\mathcal{P} -   \mathcal{P}( e^{ -t \mathfrak{L}} \mathcal{L}  )\mathcal{P} \left(\frac{1 - e^{-\mathfrak{L} t}}{\mathfrak{L}}  \mathcal{L} \right)\mathcal{P}\\
		=&\mathcal{P} \mathcal{L}(t) \mathcal{Q} [\mathcal{L}_0, \; \cdot \;]^{(-1)} (\mathcal{L} - \mathcal{L}(t))\mathcal{P}.
	\end{align*}
	Thus, we obtain \eqref{eq:K12}. Formulae \eqref{eq:I12} can be obtained similarly.
\end{proof}

\section{Relaxation conditions and Bogolubov-van Hove perturbation theory}\label{sec:relaxCond}

In this section we will analyze the long-time behavior of the projected dynamics.   Eigenvalues of GKSL generator $\mathcal{L}_0$ always have non-positive real parts \cite[p. 58]{Alicki2007}. Thus, in the generic case the real parts of the eigenvalues become strictly negative. So the limit $\lim\limits_{t \rightarrow +\infty}e^{\mathcal{L}_0 t} $ exists. But Theorem \ref{th:explicitFormulae} expresses the perturbative expansions for the time-convolutionless equations in terms of $e^{-[\mathcal{L}_0, \; \cdot \;]  t } =  e^{-\mathcal{L}_0 t} \; \cdot \;  e^{\mathcal{L}_0 t}$ which has no limit for $ t \rightarrow +\infty $ in the generic case. So to guarantee the well-define long-time behavior one should assume that terms $e^{-\mathcal{L}_0 t}$, which have exponentially growing contributions, are canceled. To formalize it we introduce relaxation conditions and enhanced relaxation conditions above. Let us for simplicity assume from now on that $t_0 =0$.

\begin{definition}
	Let the following limit
	\begin{equation}\label{eq:limitLambda}
		\lim\limits_{t \rightarrow +\infty}e^{\mathcal{L}_0 t} = \Lambda
	\end{equation}
	exist and for any $k = 1, \ldots, K$ the following equality  be satisfied
	\begin{equation}\label{eq:relaxCond}
		\mathcal{P} e^{-\mathcal{L}_0 t_k} \mathcal{L} e^{\mathcal{L}_0 (t_k- t_{k-1})} \ldots  e^{\mathcal{L}_0 (t_2 -t_1)} \mathcal{L}  e^{\mathcal{L}_0 t_1}\mathcal{P} =  \mathcal{P}  \mathcal{L} e^{\mathcal{L}_0 (t_k- t_{k-1})} \ldots  e^{\mathcal{L}_0 (t_2 -t_1)} \mathcal{L}  e^{\mathcal{L}_0 t_1}\mathcal{P},
	\end{equation}
	then we say that the \textbf{\textit{relaxation conditions}} of order $K$ are satisfied.
\end{definition}

That is the relaxation conditions just mean that first $e^{-\mathcal{L}_0 t_k}$in the expression \eqref{eq:relaxCond} could be omitted   $	\mathcal{P} e^{-\mathcal{L}_0 t_k} \ldots \mathcal{P} = 	\mathcal{P} \ldots \mathcal{P}$. Satisfaction of equality $	\mathcal{P} e^{-\mathcal{L}_0 t_k}  = \mathcal{P} $ is the most  simple and natural situation, when this is the case. In Section \ref{sec:example} we provide an explicit example, where such a condition holds.

\begin{definition}
	If the relaxation conditions are satisfied and in addition
	\begin{equation*}
		\mathcal{P} e^{-\mathcal{L}_0 t_k} \mathcal{L} e^{\mathcal{L}_0 (t_k- t_{k-1})} \ldots  e^{\mathcal{L}_0 (t_2 -t_1)} \mathcal{L}  e^{\mathcal{L}_0 t_1}\mathcal{Q} =  \mathcal{P}  \mathcal{L} e^{\mathcal{L}_0 (t_k- t_{k-1})} \ldots  e^{\mathcal{L}_0 (t_2 -t_1)} \mathcal{L}  e^{\mathcal{L}_0 t_1}\mathcal{Q}
	\end{equation*}
	are satisfied, then we say that \textbf{\textit{enhanced relaxation conditions}} of order $K$ are satisfied.
\end{definition}

\begin{lemma}\label{lem:underRel}
	Under relaxation conditions of order $2$ one has
	\begin{align}
		\mathcal{K}_{1}(t) &= 	\mathcal{P}  \mathcal{L} e^{\mathcal{L}_0 t} \mathcal{P}, \label{eq:K1underRel} \\
		\mathcal{K}_{2}(t) &=  \mathcal{P} \mathcal{L} ([e^{\mathcal{L}_0 t} , \; \cdot \;] [\mathcal{L}_0, \; \cdot \; ]^{(-1)} \mathcal{L})   \mathcal{P} 	- \mathcal{P}  \mathcal{L}  e^{\mathcal{L}_0 t}  \mathcal{P} ( [\mathcal{L}_0, \; \cdot \; ]^{(-1)} \mathcal{L})  (1-  e^{\mathcal{L}_0 t} ) \mathcal{P}. \label{eq:K2underRel}
	\end{align}
	Under enhanced relaxation conditions of order $2$ one additionally has
	\begin{align}
		\mathcal{I}_{1}(t) &= 	\mathcal{P}  \mathcal{L} e^{\mathcal{L}_0 t} \mathcal{Q},
		\label{eq:I1underRel}
		\\
		\mathcal{I}_{2}(t) &=  \mathcal{P} \mathcal{L} ([e^{\mathcal{L}_0 t} , \; \cdot \;] [\mathcal{L}_0, \; \cdot \; ]^{(-1)} \mathcal{L})   \mathcal{Q} 	- \mathcal{P}  \mathcal{L}  e^{\mathcal{L}_0 t}  \mathcal{P} ( [\mathcal{L}_0, \; \cdot \; ]^{(-1)} \mathcal{L})  (1-  e^{\mathcal{L}_0 t} ) \mathcal{Q}.
		\label{eq:I2underRel}
	\end{align}
\end{lemma}

\begin{proof}
	
	Under relaxation of order $2$ formulae \eqref{eq:K12} take the form
	\begin{equation*}
		\mathcal{K}_{1}(t) =	\mathcal{P}  \mathcal{L} e^{\mathcal{L}_0 t} \mathcal{P}
	\end{equation*}
	and
	\begin{align*}
		\mathcal{K}_{2}(t) =& - \mathcal{P} e^{-\mathcal{L}_0 t} \mathcal{L}  e^{\mathcal{L}_0 t}  \mathcal{Q}  e^{-\mathcal{L}_0 t} ( [\mathcal{L}_0, \; \cdot \; ]^{(-1)} \mathcal{L})  e^{\mathcal{L}_0 t}  \mathcal{P} 
		+ \mathcal{P} e^{-\mathcal{L}_0 t} \mathcal{L}  e^{\mathcal{L}_0 t}  \mathcal{Q} ( [\mathcal{L}_0, \; \cdot \; ]^{(-1)} \mathcal{L})   \mathcal{P} \\
		=& - \mathcal{P} e^{-\mathcal{L}_0 t} \mathcal{L} ( [\mathcal{L}_0, \; \cdot \; ]^{(-1)} \mathcal{L})  e^{\mathcal{L}_0 t}  \mathcal{P} +  \mathcal{P} e^{-\mathcal{L}_0 t} \mathcal{L}  e^{\mathcal{L}_0 t}  \mathcal{P}  e^{-\mathcal{L}_0 t} ( [\mathcal{L}_0, \; \cdot \; ]^{(-1)} \mathcal{L})  e^{\mathcal{L}_0 t}  \mathcal{P}\\
		&+ \mathcal{P} e^{-\mathcal{L}_0 t} \mathcal{L}  e^{\mathcal{L}_0 t}  ( [\mathcal{L}_0, \; \cdot \; ]^{(-1)} \mathcal{L})   \mathcal{P}
		- \mathcal{P} e^{-\mathcal{L}_0 t} \mathcal{L}  e^{\mathcal{L}_0 t}  \mathcal{P} ( [\mathcal{L}_0, \; \cdot \; ]^{(-1)} \mathcal{L})   \mathcal{P} \\
		=&  -\mathcal{P} \mathcal{L} ( [\mathcal{L}_0, \; \cdot \; ]^{(-1)} \mathcal{L})  e^{\mathcal{L}_0 t}  \mathcal{P} +  \mathcal{P} \mathcal{L}  e^{\mathcal{L}_0 t}  \mathcal{P}  ( [\mathcal{L}_0, \; \cdot \; ]^{(-1)} \mathcal{L})  e^{\mathcal{L}_0 t}  \mathcal{P}\\
		&+ \mathcal{P}\mathcal{L}  e^{\mathcal{L}_0 t}  ( [\mathcal{L}_0, \; \cdot \; ]^{(-1)} \mathcal{L})   \mathcal{P}
		- \mathcal{P}  \mathcal{L}  e^{\mathcal{L}_0 t}  \mathcal{P} ( [\mathcal{L}_0, \; \cdot \; ]^{(-1)} \mathcal{L})   \mathcal{P}\\
		=& \mathcal{P} \mathcal{L} ([e^{\mathcal{L}_0 t} , \; \cdot \;] [\mathcal{L}_0, \; \cdot \; ]^{(-1)} \mathcal{L})   \mathcal{P} 	- \mathcal{P}  \mathcal{L}  e^{\mathcal{L}_0 t}  \mathcal{P} ( [\mathcal{L}_0, \; \cdot \; ]^{(-1)} \mathcal{L})  (1-  e^{\mathcal{L}_0 t} ) \mathcal{P}.
	\end{align*}
	Thus, we have obtained \eqref{eq:K1underRel} and \eqref{eq:K2underRel}. 
	Similarly, under enhanced relaxation conditions of order 2 formulae \eqref{eq:I12} take the form \eqref{eq:I1underRel}-\eqref{eq:I2underRel}.
\end{proof}

Remark that for higher order terms  $\mathcal{K}_k(t)$ and  $\mathcal{I}_k(t)$ the relaxation conditions and enhanced relaxations, respectively, also lead to expressions which have no terms of the form $ e^{-\mathcal{L}_0 t}$. This allows us to calculate the limits of such expressions at $ t \rightarrow + \infty$ guaranteeing their existence.

\begin{lemma}\label{lem:underRelInf}
	Under relaxation conditions of order 2 one has
	\begin{align}
		\mathcal{K}_{1}(+\infty) &= \mathcal{P}  \mathcal{L} \Lambda \mathcal{P},\\
		\mathcal{K}_{2}(+\infty) &= \mathcal{P} \mathcal{L} ([\Lambda , \; \cdot \;] [\mathcal{L}_0, \; \cdot \; ]^{(-1)} \mathcal{L})   \mathcal{P} 	- \mathcal{P}  \mathcal{L}  \Lambda \mathcal{P} ( [\mathcal{L}_0, \; \cdot \; ]^{(-1)} \mathcal{L})  (1-  \Lambda ) \mathcal{P}. \label{eq:K2inf}
	\end{align}
	Under enhanced relaxation conditions of order $2$ one has
	\begin{align}
		\mathcal{I}_{1}(+\infty) &= \mathcal{P}  \mathcal{L} \Lambda \mathcal{Q},\\
		\mathcal{I}_{2}(+\infty) &= \mathcal{P} \mathcal{L} ([\Lambda , \; \cdot \;] [\mathcal{L}_0, \; \cdot \; ]^{(-1)} \mathcal{L})   \mathcal{Q}	- \mathcal{P}  \mathcal{L}  \Lambda \mathcal{P} ( [\mathcal{L}_0, \; \cdot \; ]^{(-1)} \mathcal{L})  (1-  \Lambda ) \mathcal{Q}.\label{eq:I2inf}
	\end{align}
\end{lemma}

This lemma follows immediately from Lemma \ref{lem:underRel} and condition \eqref{eq:limitLambda}.

The next theorem allows one to calculate the time-convolutionless master equations on the Bogo\-lubov-van Hove timescale $t = O(\lambda^{-2})$. From the physical point of view it is the timescale larger than bath correlation time. But in our setup the bath is another subsystem of our open composite system.

\begin{theorem}\label{th:linearOrderEq}
	Let the enhanced relaxation conditions of any order be satisfied and   $ \mathcal{K}_1(+\infty) =  \mathcal{I}_1(+\infty) = \mathcal{K}_3(+\infty) =  \mathcal{I}_3(+\infty) = 0 $, then for $t = O(\lambda^{-2})$, $\lambda \rightarrow 0$ one has
	\begin{equation}\label{eq:linearOrderSol}
		\mathcal{P} \rho(t) = e^{ \lambda^2 \mathcal{K}_{2}(+\infty) t} \mathcal{R} \rho(0) +\frac{ e^{ \lambda^2 \mathcal{K}_{2}(+\infty) t}  -1}{\mathcal{K}_{2}(+\infty)} \mathcal{I}_{2}(+\infty) \mathcal{Q} \rho(0) + O(\lambda^2)
	\end{equation}
	with 
	\begin{equation}\label{eq:renormailzationMap}
		\mathcal{R} \rho =\mathcal{P}  (1+ \lambda  \mathcal{L}  \mathcal{L}_0^{(-1)} (\Lambda - I) )  \rho,
	\end{equation}
	where $\mathcal{K}_{2}(+\infty)$ and $\mathcal{I}_{2}(+\infty)$ can be defined by formulae \eqref{eq:K2inf} and \eqref{eq:I2inf}, respectively.
\end{theorem}

\begin{proof}
	We will use the  asymptotic matching \cite[Section 7.2]{Binder1978}, \cite[Section 1.3]{Lagerstrom1988} to relate the asymptotic solution at $ t = O(\lambda^{-2}) $ and the  asymptotic solution at  $t = O(1)$, in particular  at $t = 0$.
	Equation \eqref{eq:masterEquation} in several first orders of perturbation theory takes the form 
	\begin{equation*}
		\frac{d}{dt}\mathcal{P} \rho(t) = \lambda  \mathcal{K}_{1}(t) \mathcal{P} \rho(t) + \lambda  \mathcal{I}_{1}(t)  \mathcal{Q} \rho(0)  + O(\lambda^2)
	\end{equation*}
	for fixed $t$. So the inner solution, i.e. the solution for  $t = O(1)$, takes the form
	\begin{equation*}
		\mathcal{P} \rho(t) = \mathcal{P} \rho(0) + \lambda \int_0^{t}\mathcal{K}_1(t')dt' \mathcal{P} \rho(0) + \lambda \int_0^{t}\mathcal{I}_1(t')dt'  \mathcal{Q} \rho(0) +  O(\lambda^2).
	\end{equation*}
	Under enhanced relaxation conditions if  $ \mathcal{K}_1(+\infty) =  \mathcal{I}_1(+\infty) = \mathcal{K}_3(+\infty) =  \mathcal{I}_3(+\infty) = 0 $ for $t = O(\lambda^{-2})$ we obtain
	\begin{equation*}
		\frac{d}{dt}\mathcal{P} \rho(t) =\lambda^2 \mathcal{K}_{2}(+\infty)\mathcal{P} \rho(t) + \lambda^2 \mathcal{I}_{2}(+\infty) \mathcal{Q} \rho(0)+ O(\lambda^4).
	\end{equation*}
	So the outer solution, i.e. the solution for  $t = O(\lambda^{-2}) $, takes the form
	\begin{equation*}
		\mathcal{P} \rho(t) = e^{ \lambda^2 \mathcal{K}_{2}(+\infty) t} R +\frac{ e^{ \lambda^2 \mathcal{K}_{2}(+\infty) t}  -1}{\mathcal{K}_{2}(+\infty)} \mathcal{I}_{2}(+\infty)  \mathcal{Q} \rho(0) + O(\lambda^2),
	\end{equation*}
	where $R$ is some constant matrix. For asymptotic matching one should equate the outer solution for $t = O(1)$
	\begin{equation*}
		\mathcal{P} \rho(t) =R + O(\lambda^2)
	\end{equation*}
	and the inner solution for $t = O(\lambda^{-2})$
	\begin{equation*}
		\mathcal{P} \rho(t) = \mathcal{P} \rho(0) + \lambda \int_0^{\infty}\mathcal{K}_1(t')dt' \mathcal{P} \rho(0) + \lambda \int_0^{\infty}\mathcal{I}_1(t')dt'  \mathcal{Q} \rho(0) +  O(\lambda^2),
	\end{equation*}
	which leads to	
	\begin{equation*}
		R =\mathcal{P} \rho(0) +  \lambda \left( \int_0^{+ \infty}\mathcal{K}_1(t')dt' \mathcal{P} \rho(0) +  \int_0^{+ \infty}\mathcal{I}_1(t')dt'  \mathcal{Q} \rho(0) \right).
	\end{equation*}
	Taking into account formulae \eqref{eq:K1underRel} and \eqref{eq:I1underRel} we have
	\begin{align*}
		R &=\mathcal{P} \rho(0) +  \lambda \left( \int_0^{+ \infty} dt' \mathcal{P}  \mathcal{L} e^{\mathcal{L}_0 t'} \mathcal{P} \rho(0) +  \int_0^{+ \infty} dt' \mathcal{P}  \mathcal{L} e^{\mathcal{L}_0 t'}  \mathcal{Q} \rho(0) \right)\\
		&= \mathcal{P} \rho(0) +  \lambda   \mathcal{P}  \mathcal{L}	\int_0^{+ \infty} dt'  e^{\mathcal{L}_0 t'}  \rho(0).
	\end{align*}
	As
	\begin{equation*}
		\int_0^{+ \infty} dt'  e^{\mathcal{L}_0 t'}  = \lim\limits_{t \rightarrow \infty}\frac{e^{\mathcal{L}_0 t} - 1}{\mathcal{L}_0} = \mathcal{L}_0^{(-1)} (\Lambda - I),
	\end{equation*}
	where $ \mathcal{L}_0^{(-1)}$ can be understood as in \eqref{eq:pseudoInverse} taking into account that $I -\Lambda$ is the projector on the image of $\mathcal{L}_0$, but similarly to \eqref{eq:inverseIndep} the hole expression is independent of the particular choice of the pseudoinverse, then we have
	\begin{equation*}
		R =\mathcal{P} \rho(0) + \lambda \mathcal{P}  \mathcal{L}  \mathcal{L}_0^{(-1)} (\Lambda - I)  \rho(0) =\mathcal{P}  (1+ \lambda  \mathcal{L}  \mathcal{L}_0^{(-1)} (\Lambda - I) )  \rho(0).
	\end{equation*}
	Thus, we obtain \eqref{eq:linearOrderSol} with $\mathcal{R}$ defined by \eqref{eq:renormailzationMap}.
\end{proof}

Remark that such a derivation of the renormalization operator suggests that in higher orders of perturbation theory it should coincide with perturbative expansion of $\mathcal{P}\mathcal{U}_{0}^{+\infty}$, i.e. with an open system analog of the Moller operator arising in scattering theory \cite{Derezinski2010}.

Very often one is interested only in the Bogolubov-van Hove limit without corrections and only in the case of the initial conditions consistent with the projector, i.e. the initial conditions are such that $ \mathcal{P} \rho(0) = \rho(0)$. Then one could use the following theorem.
\begin{theorem}
	Let the relaxation conditions of any order be satisfied and   $ \mathcal{K}_1(+\infty) = 0  $ and $\mathcal{Q} \rho(0)=0$, then for $t = O(\lambda^{-2})$ one has
	\begin{equation*}
		\mathcal{P} \rho(t) = e^{ \lambda^2 \mathcal{K}_{2}(+\infty) t} \mathcal{P} \rho(0) + O(\lambda).
	\end{equation*}
\end{theorem}
The proof is similar to the one of Theorem \ref{th:linearOrderEq}. But there is no inhomogeneity in the time-convolutionless master equation for  the initial conditions consistent with the projector, so one does not need to use the enhanced relaxation conditions in such a case. And also the first order corrections to the Bogolubov-van Hove limit are omitted.

\section{Example}\label{sec:example}

Let us consider an example from \cite[Section 5]{Teretenkov2019}. In the notation of this work it takes the form
\begin{align}
	\mathcal{L}_0 &= \gamma \left(|0 \rangle \langle 2| \cdot |2 \rangle \langle 0| - \frac12\{ |2 \rangle \langle 2|, \;  \cdot \;\} \right), \nonumber\\
	\mathcal{L}  &= - i g [ |2\rangle \langle 1| + | 1 \rangle \langle2|, \; \cdot \; ], \nonumber\\
	\mathcal{P} \rho &= (\rho_{00} + \rho_{22}) |0 \rangle \langle 0| + \rho_{01} |0 \rangle \langle 1| + \rho_{10} |1 \rangle \langle 0| + \rho_{11} |1 \rangle \langle 1|. \label{eq:projExample}
\end{align}

First, let us illustrate the results of Section \ref{sec:Explicit} with this example. The generator has the weak coupling type form for a generic Hamiltonian, so it is easy to solve it explicitly
\begin{align*}
	e^{\mathcal{L}_0 t} \rho =& e^{-\gamma t} \rho_{22} |2 \rangle \langle 2| + (\rho_{00} -  e^{-\gamma t} \rho_{22})|0 \rangle \langle 0| \\
	&+ \sum_{j=0}^1 e^{-\frac{\gamma}{2} t} (\rho_{2j} |2 \rangle \langle j| + \rho_{j2} |j \rangle \langle 2|) +  \rho_{01}|0 \rangle \langle 1| + \rho_{10}|1 \rangle \langle 0|  +  \rho_{11}|1 \rangle \langle 1|.
\end{align*}

Taking the limit we have
\begin{equation*}
	\lim\limits_{t \rightarrow \infty } 	e^{\mathcal{L}_0 t} \rho = (\rho_{00} + \rho_{22}) |0 \rangle \langle 0| + \rho_{01} |0 \rangle \langle 1| + \rho_{10} |1 \rangle \langle 0| + \rho_{11} |1 \rangle \langle 1|= \mathcal{P} \rho,
\end{equation*}
i.e. we obtain $\Lambda = \mathcal{P}$. 

Now let us check
\begin{equation}\label{eq:relaxCondPartCase}
	\mathcal{P}e^{-\mathcal{L}_0 t}  = \mathcal{P}
\end{equation}
by direct calculation
\begin{equation*}
	\mathcal{P}e^{-\mathcal{L}_0 t} = (\rho_{00} -  e^{\gamma t} \rho_{22} + e^{\gamma t} \rho_{22}) |0 \rangle \langle 0| + \rho_{01} |0 \rangle \langle 1| + \rho_{10} |1 \rangle \langle 0| + \rho_{11} |1 \rangle \langle 1| = \mathcal{P} \rho.
\end{equation*}

Let us remark that the similar condition holds for some other projectors and physical models, for example the ones in \cite{Reiter2012}.  Equation \eqref{eq:relaxCondPartCase} leads to the fact that the enhanced relaxation conditions are satisfied. 

Similarly, let us calculate
\begin{equation*}
	e^{\mathcal{L}_0 t} \mathcal{P}\rho  =  (\rho_{00} + \rho_{22}) |0 \rangle \langle 0| + \rho_{01} |0 \rangle \langle 1| + \rho_{10} |1 \rangle \langle 0| + \rho_{11} |1 \rangle \langle 1| =  \mathcal{P}.
\end{equation*}

Thus, Equation \eqref{eq:K1underRel} takes the form
\begin{align*}
	\mathcal{K}_{1}(t) \rho &=	\mathcal{P}  \mathcal{L} e^{\mathcal{L}_0 t} \mathcal{P} \rho= \mathcal{P}  \mathcal{L} \mathcal{P} \rho\\
	&=-ig \mathcal{P}[ |2\rangle \langle 1| + | 1 \rangle \langle2|,  (\rho_{00} + \rho_{22}) |0 \rangle \langle 0| + \rho_{01} |0 \rangle \langle 1| + \rho_{10} |1 \rangle \langle 0| + \rho_{11} |1 \rangle \langle 1|]\\
	&=  -ig \mathcal{P}(\rho_{10} |2 \rangle \langle 0| + \rho_{11} |2 \rangle \langle 1| - \rho_{01} |0 \rangle \langle 2| - \rho_{11} |1 \rangle \langle 2|) = 0
\end{align*}
and Equation \eqref{eq:K2underRel} takes the form
\begin{align*}
	\mathcal{K}_{2}(t) &=  \mathcal{P} \mathcal{L} ([e^{\mathcal{L}_0 t} , \; \cdot \;] [\mathcal{L}_0, \; \cdot \; ]^{(-1)} \mathcal{L})   \mathcal{P} 	+ \mathcal{P}  \mathcal{L}  e^{\mathcal{L}_0 t}  \mathcal{P} ( [\mathcal{L}_0, \; \cdot \; ]^{(-1)} \mathcal{L})  (1-  e^{\mathcal{L}_0 t} ) \mathcal{P}\\
	&= \mathcal{P} \mathcal{L} e^{\mathcal{L}_0 t} ([\mathcal{L}_0, \; \cdot \; ]^{(-1)} \mathcal{L})   \mathcal{P}  - \mathcal{P} \mathcal{L} ([\mathcal{L}_0, \; \cdot \; ]^{(-1)} \mathcal{L})   \mathcal{P}  =  - \mathcal{P} \mathcal{L} (1 -  e^{\mathcal{L}_0 t} )([\mathcal{L}_0, \; \cdot \; ]^{(-1)} \mathcal{L})   \mathcal{P}.
\end{align*}
It is not necessary to calculate $[\mathcal{L}_0, \; \cdot \; ]^{(-1)}$ in this formula, we just need to calculate $[\mathcal{L}_0, \; \cdot \; ]^{(-1)} \mathcal{L}$, which can be done by finding some solution $ \mathcal{X}$ of  the equation
\begin{equation*}
	[\mathcal{L}_0, \mathcal{X}] = \mathcal{L},
\end{equation*}
which leads to 
\begin{equation*}
	[\mathcal{L}_0, \; \cdot \; ]^{(-1)} \mathcal{L} =- \frac{2ig}{\gamma} \biggl(\{|1 \rangle \langle 2| - |2 \rangle \langle 1| ,\; \cdot \;\} + 2| 0\rangle\langle 1| \; \cdot \;  | 2\rangle\langle 0| - 2 | 0\rangle\langle 2| \; \cdot \;  | 1\rangle\langle 0|\biggr).
\end{equation*}
Then we obtain
\begin{equation*}
	\mathcal{K}_{2}(t) = \frac{4g^2}{\gamma} (1 - e^{- \frac{\gamma}{2} t}) \left(|0 \rangle \langle 1| \cdot |1 \rangle \langle 0| - \frac12 \{ |1 \rangle \langle 1|, \; \cdot \;\} \right).
\end{equation*}

Thus, for the initial conditions consistent with projector $\mathcal{P}$ one obtains the second-order master equation of the form
\begin{equation*}
	\frac{d}{dt} \rho_S(t) = \lambda^2 \frac{4g^2}{\gamma} (1 - e^{- \frac{\gamma}{2} t}) \left(|0 \rangle \langle 1| \rho_S(t) |1 \rangle \langle 0| - \frac12 \{ |1 \rangle \langle 1|, \rho_S(t)\} \right),
\end{equation*}
where $\rho_S(t) \equiv \mathcal{P}\rho(t)$, which coincides with Equation (10.53) in \cite{Breuer2002}. But there it arose in the usual setup of open quantum systems theory, i.e. for system dynamics, when the whole dynamics of the system and the reservoir was unitary, but the reservoir had an infinite number of degrees of freedom. This is due to the fact that these models are related by the pseudomode approach \cite{ Garraway96, Garraway97, Garraway97a, Teretenkov2019, Luchnikov19}.

Now let us illustrate the results of Section \ref{sec:relaxCond} with this model. As $ \mathcal{K}_{1}(t) =0$, then we obviously  have $\mathcal{K}_{1}(+\infty) = 0$. Using Lemma \ref{lem:underRelInf}  we obtain
\begin{align*}
	\mathcal{K}_{2}(+\infty) &=- \mathcal{P} \mathcal{L}( [\mathcal{L}_0, \; \cdot \; ]^{(-1)} \mathcal{L}) \mathcal{P} =  \frac{4g^2}{\gamma} \left(|0 \rangle \langle 1| \cdot |1 \rangle \langle 0| - \frac12\{ |1 \rangle \langle 1|, \; \cdot \;\}  \right),\\
	\mathcal{I}_{1}(+\infty) &= \mathcal{P}  \mathcal{L} \mathcal{P}  \mathcal{Q} = 0,\\
	\mathcal{I}_{2}(+\infty) &= \mathcal{P} \mathcal{L} ([\mathcal{P}  , \; \cdot \;] [\mathcal{L}_0, \; \cdot \; ]^{(-1)} \mathcal{L})   \mathcal{Q}	+ \mathcal{P}  \mathcal{L}   \mathcal{P} ( [\mathcal{L}_0, \; \cdot \; ]^{(-1)} \mathcal{L})   \mathcal{Q}\\
	&= \mathcal{P} \mathcal{L} \mathcal{P} [\mathcal{L}_0, \; \cdot \; ]^{(-1)} \mathcal{L})   \mathcal{Q} -  \mathcal{P} \mathcal{L} \mathcal{P} [\mathcal{L}_0, \; \cdot \; ]^{(-1)} \mathcal{L})  \mathcal{P}  \mathcal{Q} = 0.
\end{align*}
Similar, but awkward calculations lead to $\mathcal{K}_{3}(+\infty) = \mathcal{I}_{3}(+\infty) =0$. Thus, we can apply Theorem \ref{th:linearOrderEq}. To do it let us calculate
\begin{equation*}
	\mathcal{L}_0^{(-1)} =  \gamma^{-1} \left(|0 \rangle \langle 2| \cdot |2 \rangle \langle 0| - 2\{ |2 \rangle \langle 2|, \;\cdot \;\}  + 3 |2 \rangle \langle 2| \cdot |2 \rangle \langle 2|  \right),
\end{equation*}
which leads by formula \eqref{eq:renormailzationMap} to the following expression
\begin{align}
	\mathcal{R} \rho =& \mathcal{P}  (1 - \lambda  \mathcal{L}  \mathcal{L}_0^{(-1)}  \mathcal{Q} )  \rho \nonumber\\
	=& \left(\rho_{00} + \rho_{22} - \lambda \frac{2 i g}{\gamma}( \rho_{12} - \rho_{21}) \right) |0 \rangle \langle 0| + \left(\rho_{01} + \lambda \frac{2 i g}{\gamma} \rho_{02} \right) |0 \rangle \langle 1| + \left(\rho_{10} - \lambda \frac{2 i g}{\gamma} \rho_{20} \right) |1 \rangle \langle 0| \nonumber\\
	&+ \left(\rho_{11} + \lambda \frac{2 i g}{\gamma}( \rho_{12} - \rho_{21})\right) |1 \rangle \langle 1|. \label{eq:renomrmInParticularModel}
\end{align}
Due to $ \mathcal{I}_{2}(+\infty) = 0$ formula \eqref{eq:linearOrderSol} reduces to
\begin{equation*}
	\mathcal{P} \rho(t) = e^{ \lambda^2 \mathcal{K}_{2}(+\infty) t} \mathcal{R} \rho(0) + O(\lambda^2)
\end{equation*}
for $ t = O(\lambda^{-2})$, $\lambda \rightarrow +0$. Hence, for $ \rho_S(t) =\mathcal{P} \rho(t)  $ by omitting terms $ O(\lambda^2)$ we obtain the following equation
\begin{equation*}
	\frac{d}{dt} \rho_S(t) = \lambda^2 \frac{4g^2}{\gamma} \left(|0 \rangle \langle 1| \rho_S(t) |1 \rangle \langle 0| - \frac12 \{ |1 \rangle \langle 1|, \rho_S(t)\} \right)
\end{equation*}
for $ t = O(\lambda^{-2})$, $\lambda \rightarrow +0$. But the initial conditions for such an equation should be renormalized $\rho_S(+0) = \mathcal{R}\rho(0)$ using formula \eqref{eq:renomrmInParticularModel} if the initial conditions are not consistent with the projector. Without this correction we reproduce the results of \cite{Teretenkov2019}. With such a correction it is somewhat similar to \cite{Teretenkov2021, Teretenkov2023, Teretenkov2021Long}  in the sense that after bath correlation we have the GKSL equation, but with the renormalized initial conditions. What is different here is that the corrections occur from inhomogeneous terms in time-convolutionless. And what is especially interesting is that such terms do not contribute to dynamics after the bath correlation times, but the initial conditions are renormalized due to the initial period of order of the bath correlation time, where the non-consistency of the initial state with the projector contributes to the dynamics. 

It is possible to ask why we interpret such an example as a composite system. Firstly, in \cite{Teretenkov2019, Teretenkov19OnePart} it was shown, that such  an example arises as a zero- and one-particle restriction of the GKSL generator defined acting on matrices which are defined on tensor product Hilbert space. If the initial density matrix is restricted to zero- and one-particle spaces, then this is preserved during the evolution. Moreover, under such a restriction  Argyres-Kelley projector \cite{Argyres1964} the usual one for open quantum systems takes exactly form \eqref{eq:projExample} as discussed in \cite{Teretenkov2019, Teretenkov19OnePart}. Let as also remark that if the system is composite or not is a bit subjective in general \cite{Chernega2014}.

\section{Conclusions}

In this work we have derived explicit formulae for asymptotic expansion of time-convolutionless equations for composite open systems. More precisely, we have calculated all the integrals with respect to times in the terms which define such an expansion and obtained purely algebraic expressions for them. We have  introduced relaxation conditions and enhanced relaxation conditions, which provide time-convolutionless equations with the well-defined  Bogolubov-van Hove limit  for the initial condition consistent with the projector and for the general initial condition, respectively. We also have discussed the first order correction to the Bogolubov-van Hove limit and have given an example, when these corrections do not contribute to the master equations themselves, but renormalize their initial conditions.

We think that our results could be interesting for analysis of projection-based derivation of master equations in the usual open systems setup as well. As here we deal just with finite dimensional equations, then  we have much more hope both to derive the mathematically strict results in our setup and to verify it by numerical or symbolical calculations, than for the case of unitary dynamics, but the reservoir with infinite degrees of freedom. Nevertheless, we think  that under relaxation conditions the finite dimensional GKSL equations for the open composite systems capture main features of the usual unitary setup in the case, when the bath correlation functions of the reservoir are decaying fast enough to provide well-defined terms in all the orders of the Bogolubov-van Hove perturbation theory. 

Thus, our setup could be used as a simplified playground for more complicated open quantum systems setups. For example, it is interesting to prove strictly the results from \cite{Trushechkin2021} using the Bogolubov-van Hove perturbation theory in our setup. Moreover, possibly the results of the present paper could give a more explicit formula at least in our setup. It might be easier to analyze a perturbative expansion of steady states in such a setup, which is also interesting from the point of view of the usual open systems setup \cite{Trushechkin2022, Latune2022}. Also, the precision testing of the approaches, based on effective generators \cite{Trubilko2019, Trubilko2020, Basharov2021, Teretenkov2022Effective}, might be simplified in our setup.

As some other directions for further study, we should mention the problem of giving explicit description of generators for which relaxation conditions are satisfied. It is also interesting to analyze the multi-time correlations in  our setup. In particular, it is interesting  to check the validity of generalized regression formulae   in different orders of the Bogolubov-van Hove perturbation theory due to modern discussion in such a direction  \cite{Teretenkov2019, Teretenkov19OnePart, Lonigro2022, Chruscinski2023}. It would also be interesting to investigate other approaches taking into account non-factorized initial states in our setup \cite{Trevisan2021}. 

	\section{Acknowledgments}
	The authors thank H. Sh. Meretukov and R. Singh for the fruitful discussions of the problems considered in this work.


%
%


\begin{thebibliography}{99}
	\bibitem{Breuer2002}	
	
	H.-P.~Breuer, F.~Petruccione, \emph{The theory of open quantum systems} (Oxford University Press, Oxford, 2002).
	
	\bibitem{Saideh2020} 
	
	I.~Saideh, D.~Finkelstein-Shapiro, T.~Pullerits, A.~Keller, \textquotedblleft Projection-based adiabatic elimination of bipartite open quantum systems,\textquotedblright\;Physical Review A \textbf{102} (3), 032212 (2020).
	
	\bibitem{Finkelstein-Shapiro2020} 
	
	D.~Finkelstein-Shapiro, D.~Viennot, I.~Saideh, T.~Hansen, T.~Pullerits, A.~Keller, \textquotedblleft Adiabatic elimination and subspace evolution of open quantum systems,\textquotedblright\;Physical Review A \textbf{101} (4), 042102 (2020).
	
	\bibitem{Regent2023}
	
	F.M.~Le~Regent, P.~Rouchon, \textquotedblleft Adiabatic elimination for composite open quantum systems: Heisenberg formulation and numerical simulations,\textquotedblright\;arXiv preprint arXiv:2303.05089 (2023).
	
	\bibitem{Regent2023a}
	
	F.M.~Le~Regent, P.~Rouchon, \textquotedblleft Heisenberg formulation of adiabatic elimination for open quantum systems with two time-scales,\textquotedblright\;arXiv preprint arXiv:2303.17308 (2023).
	
	\bibitem{Arefeva2017} 
	
	I.Ya.~Aref'eva, I.V.~Volovich, O.V.~Inozemcev, \textquotedblleft Holographic control of information and dynamical topology change for composite open quantum systems,\textquotedblright\;TMP\textbf{193} (3): 1834--1843 (2017).
	
	\bibitem{Weinberg2016}
	
	S.~Weinberg, \textquotedblleft Lindblad decoherence in atomic clocks,\textquotedblright\;Physical Review A, \textbf{94} (4), 042117 (2016).
	
	\bibitem{Nakajima1958}
	
	S.~Nakajima, \textquotedblleft On Quantum Theory of Transport Phenomena: Steady Diffusion,\textquotedblright\;Progress of Theor. Phys. \textbf{20} (6), 948--959 (1958).
	
	\bibitem{Zwanzig1960}
	
	R.~Zwanzig, \textquotedblleft Ensemble Method in the Theory of Irreversibility,\textquotedblright\;J. Chem. Phys. \textbf{33} (5), 1338--1341 (1960).	
	
	\bibitem{Fulinski1967}
	
	A.~Fulinski, \textquotedblleft On the ''memory'' properties of generalized master equations,\textquotedblright\;Physics Letters A \textbf{24} (1), 63--64 (1967).
	
	\bibitem{Shibata1977}
	
	F.~Shibata, Y.~Takahashi, and N.~Hashitsume, \textquotedblleft A generalized stochastic Liouville equation. Non-Markovian versus memoryless master equations,\textquotedblright\;J. Stat. Phys., \textbf{17}, 171 (1977).
	
	\bibitem{Breuer1999}
	
	H.-P.~Breuer, B.~Kappler, and F.~Petruccione, \textquotedblleft Stochastic wave-function method for non-Markovian quantum master equations,\textquotedblright\;Phys. Rev. A, \textbf{59} (2), 1633-1643 (1999).
	
	\bibitem{Breuer2001}
	
	H.-P.~Breuer, B.~Kappler, and F.~Petruccione, \textquotedblleft The Time-convolutionless Projection Operator Technique in the Quantum Theory of Dissipation and Decoherence,\textquotedblright\;Annals of Physics, \textbf{291} (1), 36-70 (2001).
	
	\bibitem{Semin2020}
	
	V.~Semin and F.~Petruccione, \textquotedblleft Dynamical and thermodynamical approaches to open quantum systems,\textquotedblright\;Scientific reports, \textbf{10} (1), 2607 (2020).
	
	\bibitem{Chang1993}
	
	T.-M.~Chang and J.L.~Skinner, \textquotedblleft Non-Markovian population and phase relaxation and absorption lineshape for a two-level system strongly coupled to a harmonic quantum bath,\textquotedblright\;Chem. Phys., \textbf{193} (3-4), 483-539 (1993).
	
	\bibitem{Filippov2022}
	
	S.~Filippov, \textquotedblleft Multipartite correlations in quantum collision models,\textquotedblright\;Entropy, \textbf{24} (4), 508 (2022).
	
	
	\bibitem{Accardi2002}
	
	L.~Accardi, Y.G.~Lu, and I.~Volovich, \emph{Quantum theory and its stochastic limit} (Springer, Berlin, 2002).
	
	\bibitem{Davies1974}
	
	E.B.~Davies, \textquotedblleft Markovian master equations,\textquotedblright\;Commun. Math. Phys., \textbf{39} (2), 91-110 (1974).
	
	\bibitem{Accardi1990}
	
	L. Accardi, A. Frigerio, and Y. G. Lu, \textquotedblleft The weak coupling limit as a quantum functional central limit,\textquotedblright\; Commun. Math. Phys. \textbf{131} (3), 537--570 (1990).
	
	\bibitem{Bogoliubov1946}
	
	N.N.~Bogoliubov, \emph{Problems of a Dynamical Theory in Statistical Physics} (Gostekhizdat, Moscow, 1946) [in Russian].
	
	\bibitem{VanHove1954} 
	
	L.~Van Hove, \textquotedblleft Quantum-mechanical perturbations giving rise to a statistical transport equation,\textquotedblright\;Physica, \textbf{21} (1-5), 517-540 (1954).
	
	\bibitem{Kubo1963}
	
	R.~Kubo, \textquotedblleft Stochastic Liouville equations,\textquotedblright\;Journal of Mathematical Physics, \textbf{4} (2), 174-183 (1963).
	
	\bibitem{VanKampen1974}
	
	N.G.~Van~Kampen, \textquotedblleft A cumulant expansion for stochastic linear differential equations. I,\textquotedblright\;Physica, \textbf{74} (2), 215-238 (1974).
	
	\bibitem{VanKampen1974a}
	
	N.G.~Van~Kampen, \textquotedblleft A cumulant expansion for stochastic linear differential equations. II,\textquotedblright\;Physica, \textbf{74} (2), 239-247 (1974).
	
	\bibitem{Chaturvedi1979}
	
	S.~Chaturvedi and F.~Shibata, \textquotedblleft Time-convolutionless projection operator formalism for elimination of fast variables. Applications to Brownian motion,\textquotedblright\;Z. Phys. B, \textbf{35}, 297 (1979).
	
	\bibitem{Shibata1980}
	
	F.~Shibata and T.~Arimitsu, \textquotedblleft Expansion Formulas in Nonequilibrium Statistical Mechanics,\textquotedblright\;J. Phys. Soc. Jpn., \textbf{49}, 891 (1980).
	
	\bibitem{Nestmann2019} 
	
	K.~Nestmann, C.~Timm, \textquotedblleft Time-convolutionless master equation: Perturbative expansions to arbitrary order and application to quantum dots,\textquotedblright\;arXiv preprint arXiv:1903.05132 (2019).
	
	\bibitem{Teretenkov2022}
	
	A.E.~Teretenkov, \textquotedblleft Effective Gibbs State for Averaged Observables,\textquotedblright\;Entropy, \textbf{24} (8), 1144-22 (2022).
	
	\bibitem{Coddington1955}
	
	E.A.~Coddington and N.~Levinson, \emph{Theory of Ordinary Differential Equations} (McGraw-Hill, New York, 1955).
	
	\bibitem{Gorini1976}
	 
	V.~Gorini, A.~Kossakowski, E.C.G.~Sudarshan, \textquotedblleft Completely positive dynamical semigroups of N-level systems\textquotedblright\;J. of Math. Phys., \textbf{17} (5), 821--825 (1976).
	
	\bibitem{Lindblad1976}
	 
	G.~Lindblad, \textquotedblleft On the generators of quantum dynamical semigroups,\textquotedblright\;Comm. in Math. Phys., \textbf{48} (2), 119--130 (1976).
	
	\bibitem{Alicki2007}
	
	R.~Alicki and K.~Lendi, \emph{Quantum Dynamical Semigroups and Applications} (Springer, Berlin, 2007).
	
	\bibitem{Binder1978} 
	
	C.M.~Bender and S.A.~Orszag, \emph{Advanced Mathematical Methods for Scientists and Engineers}  (McGraw-Hill, New York, 1978).
	
	\bibitem{Lagerstrom1988} 
	
	P.A.~Lagerstrom, \emph{Matched Asymptotic Expansions: Ideas and Techniques} (Springer, New York (1988).
	
	\bibitem{Derezinski2010}
	J.~Derezinski, \textquotedblleft Scattering in nonrelativistic quantum field theory,\textquotedblright\;in Mathematical Horizons for Quantum Physics, pp. 147-180 (2010).
	
	\bibitem{Teretenkov2019}
	
	A.~E. Teretenkov, \textquotedblleft Pseudomode Approach and Vibronic Non-Markovian Phenomena in Light-Harvesting Complexes,\textquotedblright\;Proc. Steklov Inst. Math. \textbf{306}, 242--256 (2019).
	
	
	\bibitem{Garraway96}
	
	B.~M. Garraway and P.~L. Knight, \textquotedblleft Cavity modified quantum beats,\textquotedblright\;Phys. Rev. A, \textbf{54} (4), 3592 (1996).
	
	\bibitem{Garraway97}
	
	B.~M. Garraway, \textquotedblleft Nonperturbative decay of an atomic system in a cavity,\textquotedblright\;Phys. Rev. A \textbf{55} (3), 2290 (1997).
	
	\bibitem{Garraway97a}
	
	B.~M. Garraway, \textquotedblleft Decay of an atom coupled strongly to a reservoir,\textquotedblright Phys. Rev. A \textbf{55} (6), 4636 (1997).
	
	\bibitem{Luchnikov19}
	
	I.~A. Luchnikov, S.~V. Vintskevich, H.~Ouerdane, and S.~N. Filippov, \textquotedblleft Simulation complexity of open quantum dynamics: Connection with tensor networks,\textquotedblright\;Phys. Rev. Lett. \textbf{122} (16), 160401 (2019).
	
	\bibitem{Teretenkov2021}
	
	A.E.~Teretenkov, \textquotedblleft Non-perturbative effects in corrections to quantum master equations arising in Bogolubov-van Hove limit,\textquotedblright\;J. Phys. A, \textbf{54} (26), 265302 (2021).
	
	\bibitem{Teretenkov2023}
	
	A.E.~Teretenkov, \textquotedblleft Quantum Markovian Dynamics after the Bath Correlation Time,\textquotedblright\;Computational Mathematics and Mathematical Physics, \textbf{63} (1), 175-186 (2023).
	
	\bibitem{Teretenkov2021Long}
	
	A.E.~Teretenkov, \textquotedblleft Long-time Markovianity of multi-level systems in the rotating wave approximation,\textquotedblright\;Lobachevskii J. Math., \textbf{42} (10), 2455-2465 (2021).
	
	
	\bibitem{Teretenkov19OnePart}
	
	A.E. Teretenkov, \textquotedblleft One-particle approximation as a simple playground for irreversible quantum evolution,\textquotedblright\; Discontin. Nonlinearity Complex. \textbf{9} (4), 567--577 (2020).
	
	
	\bibitem{Argyres1964} 
	
	P.N.~Argyres and P.L.~Kelley, \textquotedblleft Theory of spin resonance and relaxation,\textquotedblright\;Physical Review, \textbf{134} (1A), A98 (1964).
	
	\bibitem{Chernega2014}
	
	V.N.~Chernega, O.V.~Man'ko, and V.I.~Man'ko, \textquotedblleft Subadditivity condition for spin tomograms and density matrices of arbitrary composite and noncomposite qudit systems,\textquotedblright\;Journal of Russian Laser Research, \textbf{35}, 278-290 (2014).
	
	\bibitem{Trushechkin2021}
	
	A.S.~Trushechkin, \textquotedblleft Derivation of the Redfield Quantum Master Equation and Corrections to It by the Bogoliubov Method,\textquotedblright\;Proc. Steklov Inst. Math., \textbf{313}, 246-257 (2021)
	
	\bibitem{Trushechkin2019}
	
	A.~Trushechkin, \textquotedblleft Calculation of coherences in Forster and modified Redfield theories of excitation energy transfer,\textquotedblright\;The Journal of Chemical Physics, \textbf{151} (7), 074101 (2019).
	
	\bibitem{Trushechkin2022}
	
	A.~Trushechkin, \textquotedblleft Quantum master equations and steady states for the ultrastrong-coupling limit and the strong-decoherence limit,\textquotedblright\;Phys. Rev. A, \textbf{106} (4), 042209 (2022).
	
	\bibitem{Latune2022}
	
	C.L.~Latune, \textquotedblleft Steady State in Ultrastrong Coupling Regime: Expansion and First Orders,\textquotedblright\;Quanta, \textbf{11} (1), 53-71 (2022).
	
	
	\bibitem{Reiter2012}
	
	F.~Reiter, A.S.~Sorensen, \textquotedblleft Effective operator formalism for open quantum systems,\textquotedblright\;Physical Review A, \textbf{85} (3), 032111 (2012).
	
	
	\bibitem{Trubilko2019}
	
	A.I.~Trubilko and A.M.~Basharov, \textquotedblleft The effective Hamiltonian method in the thermodynamics of two resonantly interacting quantum oscillators,\textquotedblright\;JETP, \textbf{129}, 339-348 (2019).
	
	\bibitem{Trubilko2020}
	
	A.I.~Trubilko and A.M.~Basharov, \textquotedblleft Hierarchy of times of open optical quantum systems and the role of the effective Hamiltonian in the white noise approximation,\textquotedblright\;JETP Letters, \textbf{111}, 532-538 (2020).
	
	\bibitem{Basharov2021}
	A.M.~Basharov, \textquotedblleft The effective Hamiltonian as a necessary basis of the open quantum optical system theory,\textquotedblright\;Journal of Physics: Conference Series, \textbf{1890},  012001 (2021).
	
	\bibitem{Teretenkov2022Effective}
	
	A.E.~Teretenkov, \textquotedblleft Effective Heisenberg equations for quadratic Hamiltonians,\textquotedblright\;Int. J. Mod. Phys. A, \textbf{37} (20-21), 243020 (2022).
	
	\bibitem{Lonigro2022}
	
	D.~Lonigro and D.~Chruscinski, \textquotedblleft Quantum regression in dephasing phenomena,\textquotedblright\;Journal of Physics A: Mathematical and Theoretical, \textbf{55} (22), 225308 (2022).
	
	\bibitem{Chruscinski2023}
	
	D.~Chruscinski, S.~Hesabi, and D.~Lonigro, \textquotedblleft On Markovianity and classicality in multilevel spin–boson models,\textquotedblright\;Scientific Reports, \textbf{13} (1), 1518 (2023).
	
	\bibitem{Trevisan2021}
	
	A.~Trevisan, A.~Smirne, N.~Megier, and B.~Vacchini, \textquotedblleft Adapted projection operator technique for the treatment of initial correlations,\textquotedblright\;Physical Review A, \textbf{104} (5), 052215 (2021).
	
	
	
\end{thebibliography}
\end{document}